%% file: main.tex
\documentclass[letterpaper,11pt]{article}

\usepackage{diagbox}
\usepackage{makecell}
\usepackage{booktabs}
\usepackage{tikz}
\usepackage{amsmath}
\usepackage{amsfonts}
\usepackage{amssymb}
\usepackage{amsthm}
\usepackage{url,ifthen}
\usepackage{enumitem}
\usepackage{srcltx}
\usepackage{multirow}
\usepackage{boxedminipage}
\usepackage[margin=1in]{geometry}
\usepackage{nicefrac}
\usepackage{xspace}
\usepackage{graphicx}
\usepackage{color}
\usepackage{subfigure}
\usepackage{colortbl}
\usepackage{setspace}
\usepackage{natbib}
\usepackage{dsfont}
\usepackage{authblk}
\usepackage{soul}
\usepackage{wrapfig}
\setcitestyle{square, authoryear}
\definecolor{DarkGreen}{rgb}{0.1,0.5,0.1}
\definecolor{DarkRed}{rgb}{0.5,0.1,0.1}
\definecolor{DarkBlue}{rgb}{0.1,0.1,0.5}
\definecolor{Gray}{rgb}{0.2,0.2,0.2}

\usepackage{listings}
\lstdefinestyle{mystyle}{
    commentstyle=\color{DarkBlue},
    keywordstyle=\color{DarkRed},
    numberstyle=\tiny\color{Gray},
    stringstyle=\color{DarkGreen},
    basicstyle=\footnotesize,
    breakatwhitespace=false,         
    breaklines=true,                 
    captionpos=b,                    
    keepspaces=true,                 
    numbers=left,                    
    numbersep=5pt,                  
    showspaces=false,                
    showstringspaces=false,
    showtabs=false,                  
    tabsize=2
}
\usepackage[small]{caption}
\usepackage[pdftex]{hyperref}
\hypersetup{
    unicode=false,          
    pdftoolbar=true,        
    pdfmenubar=true,        
    pdffitwindow=false,      
    pdfnewwindow=true,      
    colorlinks=true,       
    linkcolor=DarkBlue,          
    citecolor=DarkGreen,        
    filecolor=DarkRed,      
    urlcolor=DarkBlue,          
    pdftitle={},
    pdfauthor={},
}

\usepackage[utf8]{inputenc}
\usepackage[T1]{fontenc}
\usepackage{kpfonts}
\usepackage{microtype}

\usepackage{macros}

\usepackage{pifont}
\newcommand{\cmark}{\ding{51}\hfill}
\newcommand{\xmark}{\ding{55}\hfill}

\DeclareMathOperator{\CS}{CS}
\DeclareMathOperator{\TS}{TS}
\DeclareMathOperator{\PS}{PS}
\newcommand{\PSGap}{\PS^*_\text{diff}}
\newcommand{\TSGap}{\TS^*_\text{diff}}
\newcommand{\CSGap}{\CS^*_\text{diff}}
\newcommand{\PSRatio}{\PS^*_\text{ratio}}
\newcommand{\TSRatio}{\TS^*_\text{ratio}}
\newcommand{\CSRatio}{\CS^*_\text{ratio}}

\newtheorem{theorem}{Theorem}[section]

\newtheorem{lemma}[theorem]{Lemma}
\newtheorem{proposition}[theorem]{Proposition}

\theoremstyle{definition}
\newtheorem{definition}{Definition}[section]

\theoremstyle{remark}
\newtheorem*{remark}{Remark}

\usepackage{soul}
\usepackage{threeparttable}
\usepackage{makecell}
\usepackage{tcolorbox}
\usepackage{enumitem}
\usepackage{caption}
\usepackage{adjustbox}

\newenvironment{itm}
{\begin{itemize}[wide,noitemsep,topsep=0pt,parsep=0pt,partopsep=0pt]}
{\end{itemize}}
\newenvironment{enum}
{\begin{enumerate}[wide,noitemsep,topsep=0pt,parsep=0pt,partopsep=0pt]}
{\end{enumerate}}

\title{Regulatory Instruments for Fair Personalized Pricing}

\author[1]{Renzhe Xu}
\author[1]{Xingxuan Zhang}
\author[1*]{Peng Cui}
\author[2]{Bo Li}
\author[1]{Zheyan Shen}
\author[1]{Jiazheng Xu}

\affil[1]{Department of Computer Science and Technology, Tsinghua University, Beijing, China}
\affil[2]{School of Economics and Management, Tsinghua University, Beijing, China}
\affil[ ]{}
\affil[ ]{{\small xrz199721@gmail.com, xingxuanzhang@hotmail.com, cuip@tsinghua.edu.cn, libo@sem.tsinghua.edu.cn, \{shenzy17, xujz18\}@mails.tsinghua.edu.cn}}

\date{}

\begin{document}

\maketitle

\input{paragraphs/abstract}

\input{paragraphs/intro}
\input{paragraphs/related_works}
\input{paragraphs/preliminaries}
\input{paragraphs/impacts}
\input{paragraphs/experiments}
\input{paragraphs/discussions}
\input{paragraphs/acknowledgments}

\bibliographystyle{plainnat}
\bibliography{references}

\clearpage
\appendix
\input{paragraphs/appendix}

\end{document}

%% file: paragraphs/abstract.tex
\begin{abstract}
    \renewcommand{\thefootnote}{\fnsymbol{footnote}}
    \footnotetext[1]{Corresponding Author}
    Personalized pricing is a business strategy to charge different prices to individual consumers based on their characteristics and behaviors. It has become common practice in many industries nowadays due to the availability of a growing amount of high granular consumer data. The discriminatory nature of personalized pricing has triggered heated debates among policymakers and academics on how to design regulation policies to balance market efficiency and equity. In this paper, we propose two sound policy instruments, \textit{i.e.}, capping the range of the personalized prices or their ratios. We investigate the optimal pricing strategy of a profit-maximizing monopoly under both regulatory constraints and the impact of imposing them on consumer surplus, producer surplus, and social welfare. We theoretically prove that both proposed constraints can help balance consumer surplus and producer surplus at the expense of total surplus for common demand distributions, such as uniform, logistic, and exponential distributions. Experiments on both simulation and real-world datasets demonstrate the correctness of these theoretical results\footnote[2]{https://github.com/windxrz/fair-pricing}. Our findings and insights shed light on regulatory policy design for the increasingly monopolized business in the digital era.
\end{abstract}

%% file: paragraphs/intro.tex
\section{Introduction}

Personalized pricing, once considered the idealized construction of economic theories, has become common practice in many industries due to the availability of the increasing amount of consumer data \citep{kallus2021fairness}. 
With the high granular data of consumers' characteristics, companies can precisely assess consumers' willingness to pay and develop pricing strategies appropriately. 
The main concern of personalized pricing is that it transfers value from consumers to shareholders, increasing inequality and inefficiency from a utilitarian standpoint \citep{whitehouse2015big}.
As a result, effective regulatory policies are required to balance the benefits between consumers and companies.

The discriminatory nature of personalized pricing has triggered heated debates among policymakers and academics on designing regulatory policies to balance market efficiency and equity \citep{whitehouse2015big,gee2018fair,Brian2018are,gerlick2020ethical,gillis2020false}. 
Although several legal constraints on antitrust \citep{blair2014antitrust, united2018antitrust}, data privacy \citep{acquisti2015privacy,acquisti2016economics}, and anti-discrimination \citep{audit1997shear,kallus2021fairness} have been proposed, their impact on social welfare, especially on the balance between consumer surplus and producer surplus, remains an open question. 
Recently, \citet{dube2019personalized} have demonstrated that regulatory policies may be harmful to consumers, highlighting the challenges of developing proper policies that can guarantee the benefits of consumers.

In this paper, we study in designing effective policy instruments to balance benefits between consumers and companies.
Similar to \citep{cohen2021price}, we consider the most straightforward scenario where a monopoly sells a single product with fixed marginal cost to different consumers.
In addition, we assume that the monopoly can precisely estimate each consumers' willingness to pay and the purpose of the monopoly is to find a personalized pricing strategy to maximize its revenue while remaining compliant with the policy instruments.
We propose two sound policy instruments and prove their effectiveness in balancing consumer surplus and producer surplus.
These two policies, named $\epsilon$-difference and $\gamma$-ratio constraints, are introduced to regulate the range of personalized prices by constraining the difference and ratio between the maximal price and minimal price, respectively.

\begin{wrapfigure}{R}{0.6\textwidth}
    \centering
    \includegraphics[width=0.94\linewidth]{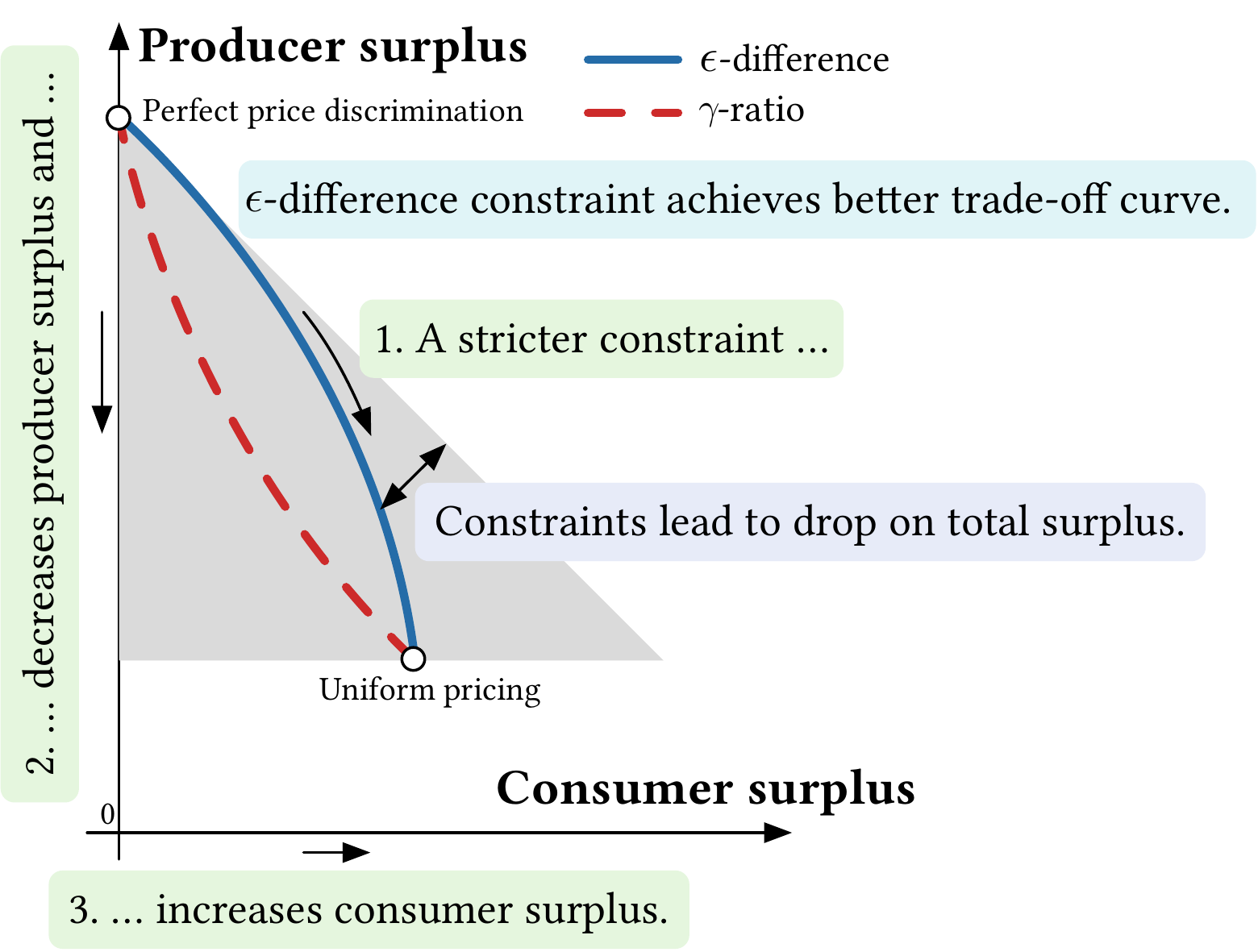}
    \caption{Graphical explanations of our major findings.}
    \figurelabel{fig:showcase}
\end{wrapfigure}

We draw conclusions on typical demand assumptions (strong regularity condition, a variant of the standard regularity \citep{myerson1981optimal}, or monotone hazard rate (MHR) \citep{hartline2013mechanism}) and common demand distributions \citep{besbes2009dynamic,besbes2015surprising,cohen2021price} (including uniform, logistic, exponential and some power law distributions).
The theoretical results are threefold.
(1) Firstly, both constraints can effectively balance the consumer surplus and producer surplus, which means the consumer surplus increases while the producer surplus decreases as the constraints become stricter.
This monotonicity property is satisfied by the full range of regulatory intensity, \textit{i.e.}, from perfect price discrimination (no regulation) \citep{mankiw2014principles} to uniform pricing (the strictest regulation) if the demand distribution is MHR or the regulatory policy is the $\gamma$-ratio constraint.
As for applying the $\epsilon$-difference constraint on strongly regular demand distributions, the property also holds for a large, despite not full, range of regulatory intensity.
(2) Secondly, we compare the trade-off between consumer surplus and producer surplus achieved by the two constraints and show that $\epsilon$-difference constraint outperforms $\gamma$-ratio constraint.
This means that the consumer surplus under $\gamma$-ratio constraint is smaller than that under $\epsilon$-difference constraint if the producer surplus under the two constraints is equal.
(3) Thirdly, imposing either of the constraints will inevitably, to some extent, harm the total surplus. This result is reasonable, given that the efficiency-equity trade-off is largely recognized in practice \citep{mankiw2014principles}.
In addition, the perfect price discrimination achieves the maximal market efficiency and any regulatory policies to avoid it will inevitably harm the total surplus.
Graphical explanations of these findings are shown in \figureref{fig:showcase}. These theories are validated by experiments on both simulation and real-world datasets \citep{wertenbroch2002measuring,slunge2015willingness,phillips2015effectiveness}.

To conclude, for industry practitioners and policymakers, our paper offers the following takeaways.
\begin{enumerate}[leftmargin=*,noitemsep,topsep=0pt,parsep=0pt,partopsep=0pt]
    \item We propose two sound and effective policy instruments on the range of personalized prices, \textit{i.e.}, the difference or ratio between the maximal price and minimal price, and study their impacts on consumer surplus, producer surplus, and social welfare.
    \item For common demand distributions, both constraints can help balance consumer surplus and producer surplus despite the expense of total surplus, which implies that they can protect consumers in the increasingly monopolized business in the digital era.
    \item Comparatively, the $\gamma$-ratio constraint is more suitable for designing policies on the grounds that ratios could be easily adapted to various scenarios. The $\epsilon$-difference constraint has better performance on the trade-off between consumer surplus and producer surplus. As a result, the two constraints could be adopted in different applications in practice.
\end{enumerate}

%% file: paragraphs/related_works.tex
\section{Related works} \sectionlabel{sect:related-works}
\paragraph{Personalized pricing and price discrimination} With the increasing amount of consumers' data, personalized pricing or price discrimination have become common practice in grocery chains \citep{clifford2012shopper}, department stores \citep{d2017neiman}, airlines \citep{tuttle2013flight}, and many other industries \citep{whitehouse2015big}. The value of personalized pricing for the companies are studied in \citep{barlow1963properties,tamuz2013lower,munoz2017revenue,elmachtoub2021value}. In addition, effective approaches have been developed to achieve personalized pricing, including both online \citep{qiang2016dynamic,javanmard2019dynamic,ban2021personalized} and offline algorithms \citep{chen2021statistical,biggs2021model}.

The price discrimination is often achieved by monopolies, \textit{i.e.}, firms that is the sole seller of a product without close substitutes \citep{mankiw2014principles}. They have absolute market power and consumers have to take the prices offered by them. There are three types of price discrimination \citep{shapiro1998information}. The most extreme case is first-degree price discrimination (or perfect price discrimination equivalently), under which circumstances the prices offered to consumers are exactly their willingness to pay. Second-degree price discrimination occurs when a company charges different prices for various quantities consumed, such as buck discounts. Third-degree price discrimination divides the market into segments and charges a different price to each segment. In this paper, we consider developing policy instruments towards first-degree price discrimination.

\paragraph{Social welfare analysis under personalized pricing}
The benefits earned by consumers, producers, and society in a market are often measured as welfare, \textit{i.e.}, consumer surplus, producer surplus (or revenue equivalently), and total surplus respectively \citep{mankiw2014principles}.

Under perfect price discrimination, consumer surplus is zero while producer surplus is maximized. Several literatures also studied the welfare implications of third-degree price discrimination. \citet{schmalensee1981output} and \citet{varian1985price} noted a necessary condition for third-degree price discrimination to increase social welfare is that output increase. \citet{bergemann2015limits} proved that an intermediary between consumers and companies who knows the distribution of consumers' exact willingness to pay can design market segments to maximize any linear combination of consumer surplus and seller revenue. \citet{cummings2020algorithmic} further studied the theoretical computational efficiency of finding such segmentation. Recently, \citet{dube2019personalized} found that finer-grained personalized pricing in third-degree price discrimination can increase consumer welfare, which is contrary to the common belief that personalized pricing will always harm consumers.

\paragraph{Fair regulation towards personalized pricing}
The discriminatory nature of personalized pricing has triggered heated debate among policymakers and academics on designing fair regulatory policies to restrict price discrimination \citep{whitehouse2015big,gee2018fair,Brian2018are,gerlick2020ethical,gillis2020false}. People have developed legal constraints on antitrust \citep{blair2014antitrust, united2018antitrust}, data privacy \citep{acquisti2015privacy,acquisti2016economics}, and anti-discrimination \citep{audit1997shear,kallus2021fairness}. The former two constraints can help mitigate price discrimination by avoiding the formation of monopolies and precise estimations of consumers' willingness to pay. Anti-discrimination constraints aim to protect different subgroups of consumers such as female and blacks. In this paper, similar to \citep{cohen2021price} we suppose a monopoly with perfect information of consumers' willingness to pay and consider the regulation towards first-degree price discrimination, which differs from settings considered by the legal constraints mentioned above.

Several constraints on pricing algorithms are also considered. \citet{li2016behavior} studied the impact of consumers’ fairness concerns in a duopoly market. \citet{kallus2021fairness,cohen2021price} proposed several fair pricing constraints for anti-discrimination. \citet{biggs2021model} presented a customized, prescriptive tree-based algorithm and the depth of the tree can be considered as a restriction on the segmentation granularity.

\paragraph{Fair machine learning} There are increasing concerns on fairness~\citep{dwork2012fairness,mehrabi2021survey} and robustness~\citep{zhang2021deep,zhang2021domain,shen2021towards} recently. Various fairness notions, including group fairness \citep{hardt2016equality,kearns2018preventing,xu2020algorithmic}, individual fairness \citep{yurochkin2020sensei,dwork2012fairness}, and causality-based fairness notions \citep{kilbertus2017avoiding,kusner2017counterfactual,chiappa2019path} are proposed to protect different subgroups or individuals.
Recently, several works \citep{heidari2019fairness,hu2020fair,rolf2020balancing,kasy2021fairness} have connected fairness with welfare analysis in allocating decisions.
A thorough survey on fair machine learning can be found in \citep{mehrabi2021survey}.

%% file: paragraphs/preliminaries.tex
\section{Preliminaries} \sectionlabel{sect:preliminaries}

\subsection{Notations}
We consider a single-period setting where a monopoly offers a single product, with fixed marginal cost $c \ge 0$ to different consumers. Let $V$ denote the consumers' willingness to pay. $V$ is supported on $[0, U]$ ($U$ could be $\infty$, and in this case $V$ is supported on $[0, \infty)$) and is drawn independently from a distribution $F$, called demand distribution. Let $f(v)$ be the probability density function, $F(v)$ be the cumulative density function, \textit{i.e.}, $F(v) = \mathbb{P}[V \le v]$. In addition, we introduce the survival function $S(v) \triangleq 1 - F(v)$ and the hazard rate function $h(v) \triangleq f(v)/S(v)$. $S(v)$ and $f(v)$ can be expressed with $h(v)$ as
\begin{equation}
    S(v) = \exp \left(-\int_0^v h(t)\mathrm{d}t\right), \quad f(v) = h(v)\exp \left(-\int_0^v h(t)\mathrm{d}t\right).
\end{equation}

We suppose the monopoly could precisely estimate consumers' willingness to pay and make personalized prices accordingly. The pricing function is defined as $p: [0,U] \rightarrow [0, U]$, which means a consumer with willingness to pay $V$ is charged with price $p(V)$. We assume that the monopoly has enough supply to fulfill all the demand. In addition, a consumer with willingness to pay $V$ buys the product only if $V$ is at least the offered price $p(V)$, \textit{i.e.}, the quantity demanded of the consumer is $\mathbb{I}[p(V) \le V]$. Therefore, $S(v)$ can be considered as the overall quantity demanded function at price $v$.

Under pricing strategy $p$, the benefits of monopolies, consumers, and society as a whole are measured by the producer surplus $\PS(p)$ (or revenue equivalently), consumer surplus $\CS(p)$, and total surplus $\TS(p)$ respectively and are given by
\begin{equation}
    \left\{
    \begin{aligned}
        \PS(p) & = \mathbb{E}[\mathbb{I}[p(V) \le V](p(V) - c)], \\
        \CS(p) & = \mathbb{E}[\mathbb{I}[p(V) \le V](V - p(V))], \\
        \TS(p) & = \mathbb{E}[\mathbb{I}[p(V) \le V](V - c)].
    \end{aligned}
    \right.
\end{equation}

\subsection{Assumptions on willingness to pay}
To model the distribution of consumers' willingness to pay, We adopt regular and monotone hazard rate distributions from auction theory and revenue management \citep{myerson1981optimal,yan2011mechanism,babaioff2012dynamic,hartline2013mechanism}. For a thorough analysis on social welfare, we slightly strengthen the assumption on regularity as follows.

\begin{definition} [$k$-strongly regular distribution] \definitionlabel{defn:strongly-regular}
    We say $F$ is a $k$-strongly regular distribution if
    \begin{enumerate}
        \item $F(\cdot)$ is twice differentiable, and
        \item the function $w(v) \triangleq v - S(v)/f(v)$ is a monotone strictly increasing function, and
        \item $\lim_{v \rightarrow U} w(v) > k$.
    \end{enumerate}
    Furthermore, if $F$ is $k$-strongly regular for any $k>0$, we say $F$ is $\infty$-strongly regular.
\end{definition}

\begin{remark}
    Compared with standard regular distributions \citep{myerson1981optimal}, $k$-strongly regular distributions differ in two aspects. Firstly, we assume the strict monotonicity of $w(v)$ here while the standard one assumes the non-decreasing property. Secondly, we further assume the lower bound of the limit of $w(v)$ when $v \rightarrow U$. We add these two additional assumptions to guarantee the existence and uniqueness of the optimal pricing strategy under our proposed regulatory policies.
\end{remark}

In addition, monotone hazard rate distributions are defined as follows.
\begin{definition} [Monotone hazard rate (MHR) distribution] \definitionlabel{defn:MHR}
    We say $F$ is a monotone hazard rate distribution if $F(\cdot)$ is twice differentiable and the hazard rate $h(v) = f(v)/S(v)$ is non-decreasing.
\end{definition}

The monotone hazard rate assumption is stronger than the strongly regular assumption, given by the following proposition.

\begin{proposition} \propositionlabel{prop:mhr-and-regular}
    Suppose $F$ is supported on $[0, U]$ and is a monotone hazard rate distribution. Then $\forall k < U$, F is also a $k$-strongly regular distribution.
\end{proposition}

Both strongly regular and monotone hazard rate assumptions are mild and common demand distributions \citep{besbes2009dynamic,besbes2015surprising,cohen2021price}, such as uniform, exponential, logistic distributions, are satisfied by both of them. In addition, some power law distributions are strongly regular, despite not MHR. See \tableref{tab:common-distributions} for detailed properties of these distributions.

\begin{table}[t]
    \caption{Detailed properties of several distributions. Common demand distributions, including uniform, exponential, and logistic distributions are both monotone hazard rate (MHR) and strongly regular. Some power law distributions are also strongly regular, despite not MHR.}
    \begin{adjustbox}{max width=\textwidth}
        \begin{threeparttable}
            \begin{tabular}{c|cccc}
                \toprule
                Distribution & Uniform & Exponential & Logistic & Power law \tnote{*}\\
                \midrule
                Support and parameters & $[0, a]$. $a > 0$ & $[0, \infty)$. $\lambda > 0$ & $[0, \infty)$. $s > 0$, $\mu$ & $[0, \infty)$. $\Delta > 0$, $\alpha > 0$ \\
                Probability density function & $\frac{1}{a}$ & $\lambda \mathrm{e}^{-\lambda v}$ & $\frac{\mathrm{e}^{-(v - \mu) / s}}{s\left(1+\mathrm{e}^{-(v-\mu)/s}\right)^2}$ & $\alpha \Delta^{\alpha}(v+\Delta)^{-(\alpha+1)}$ \\
                Hazard rate function & $\frac{a}{(v-a)^2}$  & $\lambda$ & $\frac{1}{s\left(1+\mathrm{e}^{-(v-\mu)/s}\right)}$ & $\frac{\alpha}{v + \Delta}$ \\
                \midrule
                MHR? & \cmark & \cmark & \cmark & \xmark \\
                Strongly regular? & \cmark, $\forall k<a$, $k$-strongly regular & \cmark, $\infty$-strongly regular & \cmark, $\infty$-strongly regular & \cmark, $\infty$-strongly regular if $\alpha > 1$ \\
                \bottomrule
            \end{tabular}
            \begin{tablenotes}
                \item[*] To ensure the support of power law distribution is $[0, \infty)$, here we adopt the power law + shortscale distribution as shown in \citep{zang2018learning}.
            \end{tablenotes}
        \end{threeparttable}
    \end{adjustbox}
    \tablelabel{tab:common-distributions}
\end{table}

%% file: paragraphs/impacts.tex
\section{Proposed policy instruments and their impacts on social welfare} \sectionlabel{sect:impacts}
In this section, we propose two sound policy instruments and discuss the impacts of them on social welfare. The policy instruments are regulatory constraints on the range of personalized prices as shown in \sectionref{sect:constraints}, namely the $\epsilon$-difference and $\gamma$-ratio constraints. Then we prove the existence and uniqueness of the optimal pricing strategy under either of the constraints in \sectionref{sect:existence-and-uniqueness}. Afterward, we prove that both constraints can help balance the consumer surplus and producer surplus (\sectionref{sect:balancing}) at the expense of total surplus (\sectionref{sect:total-surplus}). We compare the trade-off between consumer surplus and producer surplus under the two constraints in \sectionref{sect:gap-vs-ratio}. To simplify the proofs, we suppose the marginal cost $c$ is zero in sections listed above. But the results could be applied to general settings when $c > 0$, which is shown in \sectionref{sect:marginal-cost-greater}.

\subsection{Proposed policy instruments} \sectionlabel{sect:constraints}
Without any constraint, the monopoly could charge each consumer with his or her willingness to pay exactly, which is well known as perfect price discrimination or first-degree price discrimination \citep{shapiro1998information}. In this case, consumers get no benefits and the revenue is maximized. Now we consider constraining the maximal price difference and ratio for the pricing strategy $p$.

\begin{tcolorbox}[boxsep=3pt,left=2pt,right=2pt,top=0pt,bottom=0pt]
    \begin{definition}[$\epsilon$-difference fair]
        $\forall 0 \le \epsilon < U$, we say pricing strategy $p$ is $\epsilon$-difference fair if the maximal price difference is no more than $\epsilon$, \textit{i.e.},
        \begin{equation}
            \max_v p(v) - \min_v p(v) \le \epsilon.
        \end{equation}
    \end{definition}
\end{tcolorbox}

\begin{tcolorbox}[boxsep=3pt,left=2pt,right=2pt,top=0pt,bottom=0pt]
    \begin{definition}[$\gamma$-ratio fair]
        $\forall \gamma \ge 1$, we say pricing strategy $p$ is $\gamma$-ratio fair if
        \begin{equation}
            \max_v (p(v) - c) \le \gamma \cdot \min_v (p(v) - c).
        \end{equation}
    \end{definition}
\end{tcolorbox}

\begin{remark}
    We subtract the price with the marginal cost here for normalization. The ratio constraint is well defined because the minimal price must be greater than the marginal cost from the producer's perspective. After the subtraction, the effective range of $\gamma$ is scale-free and is always $[1, \infty)$. By contrast, the range of $\epsilon$ depends on the support of the underlying demand distribution. Therefore, the setting of $\gamma$ is more generic in different applications.
\end{remark}

Let $p_l = \min_v p(v)$ be the upper price and $p_u = \max_v p(v)$ be the lower price. To maximize the revenue, $p_l$ must be greater than marginal cost $c$ and the pricing strategy must be
\begin{equation}
    p(v) =
    \begin{cases}
        p_u, & \text{if } v \ge p_u, \\
        v, & \text{if } p_l \le v < p_u, \\
        p_l, & \text{otherwise}.
    \end{cases}
\end{equation}
Hence, the optimal pricing strategy can be determined by the lower price $p_l$ and the upper price $p_u$. The corresponding producer surplus, consumer surplus, and total surplus can be written as functions of $p_l$ and $p_u$.
\begin{equation}
    \left\{
    \begin{aligned}
        \PS(p_l, p_u) & = (p_u - c)S(p_u) + \int_{p_l}^{p_u} (v - c)f(v)\mathrm{d}v, \\
        \CS(p_l, p_u) & = \int_{p_u}^{U} (v-p_u)f(v)\mathrm{d}v, \\
        \TS(p_l, p_u) & = \int_{p_l}^{U} (v - c)f(v) \mathrm{d}v.
    \end{aligned}
    \right.
\end{equation}
As a result, the optimal pricing strategy can be formulated.
\begin{itemize}[wide,labelindent=0pt,noitemsep,topsep=0pt,parsep=0pt,partopsep=0pt]
    \item The optimal $\boldsymbol{\epsilon}$\textbf{-difference fair} pricing strategy can be given as
    \begin{equation} \equationlabel{eq:gap-constraint}
        p_l^*(\epsilon), p_u^*(\epsilon) = \arg \max_{p_l, p_u} \, \PS(p_l, p_u), \quad \text{s.t.} \quad p_u - p_l \le \epsilon.
    \end{equation}
    \item The optimal $\boldsymbol{\gamma}$\textbf{-ratio fair} pricing strategy can be given as
    \begin{equation} \equationlabel{eq:ratio-constraint}
        q_l^*(\gamma), q_u^*(\gamma) = \arg \max_{q_l, q_u} \, \PS(q_l, q_u), \quad \text{s.t.} \quad  \frac{q_u - c}{q_l - c} \le \gamma.
    \end{equation}
\end{itemize}

The producer surplus, consumer surplus, and total surplus under the optimal $\epsilon$-difference fair pricing strategy are given as follows:
\begin{equation}
    \PSGap(\epsilon) \triangleq \PS(p_l^*(\epsilon), p_u^*(\epsilon)), \quad \CSGap(\epsilon) \triangleq \CS(p_l^*(\epsilon), p_u^*(\epsilon)), \quad \TSGap(\epsilon) \triangleq \TS(p_l^*(\epsilon), p_u^*(\epsilon)).
\end{equation}
Similarly, the surpluses under the optimal $\gamma$-ratio fair pricing strategy are given as:
\begin{equation}
    \PSRatio(\gamma) \triangleq \PS(q_l^*(\gamma), q_u^*(\gamma)), \quad \CSRatio(\gamma) \triangleq \CS(q_l^*(\gamma), q_u^*(\gamma)), \quad \TSRatio(\gamma) \triangleq \TS(q_l^*(\gamma), q_u^*(\gamma)).
\end{equation}

\subsection{Existence and uniqueness of the optimal pricing strategy} \sectionlabel{sect:existence-and-uniqueness}
In this subsection, we show the existence and uniqueness of the optimal pricing strategy under either of the constraints.
\subsubsection{\texorpdfstring{$\epsilon$}{Epsilon}-difference fair} \sectionlabel{sect:gap}
\begin{proposition} \propositionlabel{prop:gap-existence}
    When $c=0$, if $F$ is a $c$-strongly regular distribution, then the solution $(p_l^*(\epsilon), p_u^*(\epsilon))$ to \equationref{eq:gap-constraint} exists and is unique. In addition, $p_l^*(\epsilon)$ and $p_u^*(\epsilon)$ are differentiable.
\end{proposition}

\subsubsection{\texorpdfstring{$\gamma$}{Gamma}-ratio fair} \sectionlabel{sect:ratio}
\begin{proposition} \propositionlabel{prop:ratio-existence}
    When $c=0$, if $F$ is a $c$-strongly regular distribution, then the solution $(q_l^*(\gamma), q_u^*(\gamma))$ to \equationref{eq:ratio-constraint} exists and is unique. In addition, $q_l^*(\gamma)$ and $q_u^*(\gamma)$ are differentiable.
\end{proposition}

\subsection{Balancing consumer surplus and producer surplus} \sectionlabel{sect:balancing}
A stronger constraint can inevitably lead to the decrease in producer surplus. However, it remains a question on whether the constraints can lead to an increase in consumer surplus. To answer it, we first show that the lower price $p_l$ and upper price $p_u$ have close relationships with total surplus and consumer surplus.

\begin{proposition} \propositionlabel{prop:price-and-surplus}
    Total surplus is strictly decreasing \textit{w.r.t.} lower price $p_l$ when $p_l > c$. Consumer surplus is strictly decreasing \textit{w.r.t.} upper price $p_u$.
\end{proposition}

With this proposition, the impact of the constraints on consumer surplus can be measured by the monotonicity of the optimal upper price $p_u^*(\epsilon)$ and $q_u^*(\gamma)$.

\subsubsection{\texorpdfstring{$\epsilon$}{Epsilon}-difference fair}
\begin{theorem} \theoremlabel{thrm:gap-increasing-upper-price}
    When $c=0$, if F is a monotone hazard rate distribution, then $p_u^*(\epsilon)$ is strictly increasing and $\CSGap(\epsilon)$ is strictly decreasing \textit{w.r.t.} $\epsilon$.
\end{theorem}
\begin{remark}
    We should notice that, a stronger constraint represents a smaller $\epsilon$. As a result, according to the theorem, when $F$ is MHR, a stronger constraint will lead to a decrease in optimal upper price $p_u^*$, as well as an increase in consumer surplus.
\end{remark}

However, when $F$ is not a monotone hazard rate distribution, even if $F$ is a strongly regular distribution, we can not guarantee the monotonicity of $p_u^*(\epsilon)$. But we can show that $p_u^*(\epsilon)$ is monotone over a range of $\epsilon$.

\begin{theorem} \theoremlabel{thrm:gap-consumer-under-regular}
    When $c=0$, if $F$ is a $c$-strongly regular distribution, let $\epsilon_0$ be the solution to equation $\epsilon - 2p_l^*(\epsilon) = 0$. Then when $\epsilon > \epsilon_0$, $p_u^*(\epsilon)$ is strictly increasing and $\CSGap(\epsilon)$ is strictly decreasing \textit{w.r.t.} $\epsilon$.
\end{theorem}

\begin{remark}
    We first notice that $p_l^*(\epsilon)$ is decreasing and show it in \theoremref{thrm:gap-decreasing-lower-price}. This implies $\epsilon - 2 p_l^*(\epsilon)$ is strictly increasing. In addition, $p_l^*(\epsilon) \le p_l^*(0)$. As a result, the solution to $\epsilon - 2p_l^*(\epsilon) = 0$ exists and $\epsilon_0 \le 2 p_l^*(0)$, which means $\epsilon_0$ is not greater than the double of uniform price. We empirically show that the range of $\epsilon > \epsilon_0$ is fairly large on a class of common strongly regular but not MHR demand distributions, \textit{i.e.}, power law distributions in \sectionref{sect:experiments}.
\end{remark}

\subsubsection{\texorpdfstring{$\gamma$}{Gamma}-ratio fair}
As shown above, $\epsilon$-difference fair can not guarantee the increase of consumer surplus when demand distribution is strongly regular. However, as for the $\gamma$-ratio constraint, the monotonicity of consumer surplus will always hold.

\begin{theorem} \theoremlabel{thrm:ratio-monotone-upper-price}
    When $c=0$, if F is a $c$-strongly regular distribution, $q_u^*(\gamma)$ is strictly increasing and $\CSRatio(\gamma)$ is strictly decreasing \textit{w.r.t.} $\gamma$.
\end{theorem}
\begin{remark}
    Similar to $\epsilon$-difference fair, a stronger constraint corresponds to a smaller $\gamma$ here. As a result, a stronger constraint on $\gamma$-ratio fair will lead to the decrease in optimal upper price $q_u^*(\gamma)$ and increase in consumer surplus.
\end{remark}

\subsection{Drop on total surplus} \sectionlabel{sect:total-surplus}
In this subsection, we show that imposing either of the constraints will harm total surplus. This result is reasonable, given that the efficiency-equity trade-off is primarily recognized in practice \citep{mankiw2014principles}. In addition, the perfect price discrimination achieves the maximal market efficiency and any regulatory policy attempting to avoid perferct price discrimation will inevitably harm the total surplus. Similar to previous sections, stronger constraints represent smaller $\epsilon$ and $\gamma$ in the two constraints.

\subsubsection{\texorpdfstring{$\epsilon$}{Epsilon}-difference fair}
\begin{theorem} \theoremlabel{thrm:gap-decreasing-lower-price}
    When $c=0$, if F is a $c$-strongly regular distribution, then $p_l^*(\epsilon)$ is strictly decreasing and $\TSGap(\epsilon)$ is strictly increasing \textit{w.r.t.} $\epsilon$.
\end{theorem}

\subsubsection{\texorpdfstring{$\gamma$}{Gamma}-ratio fair}
\begin{theorem} \theoremlabel{thrm:ratio-monotone-lower-price}
    When $c=0$, if F is a $c$-strongly regular distribution, then $q_l^*(\gamma)$ is strictly decreasing and $\TSRatio(\gamma)$ is strictly increasing \textit{w.r.t.} $\gamma$.
\end{theorem}

\subsection{\texorpdfstring{$\epsilon$-difference fair vs $\gamma$-ratio fair}{Difference fair vs ratio fair}}\sectionlabel{sect:gap-vs-ratio}
We compare the performance on the trade-off of between consumer surplus and producer surplus in this subsection under either of the constraints. The result is given by the following theorem.

\begin{theorem} \theoremlabel{thrm:gap-vs-ratio}
    Suppose $F$ is strongly regular. Suppose $0 \le \epsilon < U$, $\gamma \ge 1$, and $\CSGap(\epsilon) = \CSRatio(\gamma)$. Then $\TSGap(\epsilon) \ge \TSRatio(\gamma)$ and $\PSGap(\epsilon) \ge \PSRatio(\gamma)$.
\end{theorem}
\begin{remark}
    This theorem shows that the $\epsilon$-difference constraint outperforms the $\gamma$-ratio constraint on the trade-off between consumer surplus and producer surplus, resulting in higher efficiency while fairness is guaranteed.
\end{remark}

\subsection{General cases when marginal cost is positive} \sectionlabel{sect:marginal-cost-greater}
To simplify proofs, we suppose the marginal cost $c$ is zero in previous subsections. In this subsection, we show that the conclusions also hold for general settings when $c$ is greater than zero.

The general idea is that we can transfer the raw distribution $F$ to a new distribution $\tilde{F}$ supported on $[0, U - c]$ ($U - c = \infty$ if $U = \infty$) and the social welfare for positive marginal cost on distribution $F$ are equal to that for zero marginal cost on the new distribution $\tilde{F}$. The corresponding functions can be given as
\begin{equation}
    \tilde{f}(v) = \frac{f(v + c)}{S(c)}, \quad \tilde{S}(v) = \frac{S(v + c)}{S(c)}, \quad \tilde{h}(v) = h(v + c).
\end{equation}
Then the revenue, consumer surplus, and total surplus with lower price $p_l$ and upper price $p_u$ can be written as
\begin{equation}
    \left\{
    \begin{aligned}
        \PS(p_l, p_u) & = (p_u - c)S(p_u) + \int_{p_l}^{p_u} (v - c)f(v)\mathrm{d}v = S(c)\left((p_u - c)\tilde{S}(p_u - c) + \int_{p_l - c}^{p_u - c} v \tilde{f}(v)\mathrm{d}v\right), \\
        \CS(p_l, p_u) & = \int_{p_u}^{u} (v-p_u)f(v)\mathrm{d}v = S(c)\int_{p_u - c}^{u - c} (v-(p_u - c))\tilde{f}(v)\mathrm{d}v, \\
        \TS(p_l, p_u) & = \int_{p_l}^{u} (v - c)f(v) \mathrm{d}v = S(c)\int_{p_l - c}^{u} v\tilde{f}(v) \mathrm{d}v.
    \end{aligned}
    \right.
\end{equation}
They are proportional to producer surplus, consumer surplus and total surplus with lower price $(p_l - c)$ and upper price $(p_u - c)$ when the demand distribution is $\tilde{F}$ and marginal cost is zero. In addition, the properties of $\tilde{F}$ are given by the following proposition.

\begin{proposition} \propositionlabel{prop:hold-valuation-assumption}
    If $F$ is a monotone hazard rate distribution, so does $\tilde{F}$. If $F$ is a $c$-strongly regular distribution, $\tilde{F}$ will be a $0$-strongly regular distribution.
\end{proposition}

As a result, our major conclusions in previous subsections hold for more general settings when marginal cost is positive.

%% file: paragraphs/experiments.tex
\section{Experiments} \sectionlabel{sect:experiments}
We run experiments on both simulation and real-world datasets to prove the correctness of our theoretical results.
\subsection{Simulation}
\subsubsection{Data} 
We simulate common demand distributions including MHR distributions, namely uniform and exponential distributions \citep{besbes2009dynamic,besbes2015surprising,cohen2021price}, and a strongly regular distribution, namely power law distribution \citep{zang2018learning}.
The detailed properties of these distributions are shown in \tableref{tab:common-distributions}. We choose $a=1$ for uniform distribution, $\lambda=1$ for exponential distribution, and $\Delta=1$, $\alpha=2$ for power law distribution. The marginal cost is set to $0$. Note that the choice of parameters does not bring a difference as long as the distribution is MHR or strongly regular.

\begin{figure}[t]
    \centering
    \includegraphics[width=\linewidth]{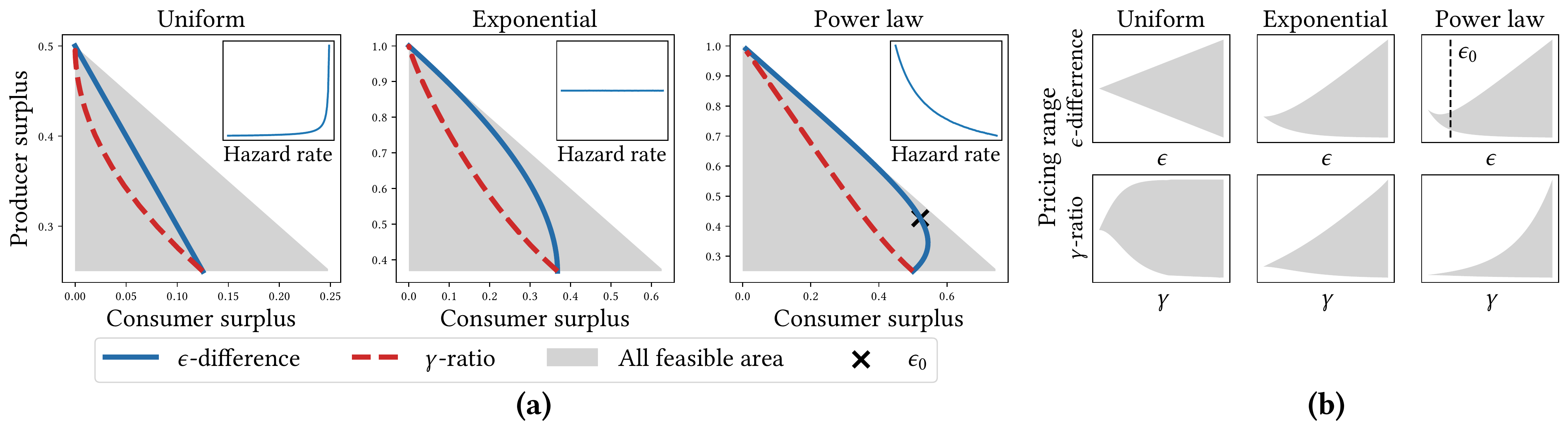}
    \caption{Experiments on simulations. Subfigure (a) shows trade-off curves between consumer surplus and producer surplus under the $\epsilon$-difference and $\gamma$-ratio constraints, as well as the hazard rate functions of uniform, exponential, and power law distributions. Subfigure (b) shows the pricing range under the constraints.}
    \figurelabel{fig:cs-and-revenue-simulation}
\end{figure}

\subsubsection{Results and analysis}
With closed-form probability density functions of these distributions, we can optimize \equationref{eq:gap-constraint} and \equationref{eq:ratio-constraint} directly. We choose different $\epsilon$ and $\gamma$ and calculate the corresponding optimal lower / upper price, consumer surplus, producer surplus and total surplus accordingly.

The results on trade-off curves between consumer surplus and producer surplus are shown in \figureref{fig:cs-and-revenue-simulation}(a). The grey area in \figureref{fig:cs-and-revenue-simulation}(a) represents all possible pairs of consumer surplus and producer surplus that: (1) consumer surplus is nonnegative, (2) producer surplus is not less than revenue under uniform pricing, and (3) total surplus does not exceed the surplus generated by efficient trade \citep{bergemann2015limits}, \textit{i.e.}, $\mathbb{E}[\mathbb{I}[V \ge c](V-c)]$. The endpoints of the trade-off curve for both constraints fall on the upper left corner and the bottom line of the grey area, which denote the perfect price discrimination and the uniform pricing, respectively. We also show the pricing range under both constraints in \figureref{fig:cs-and-revenue-simulation}(b)
. From the figure, we have the following observations and they all coincide with the theoretical results listed in \sectionref{sect:impacts}.

\begin{itm}
    \item \textit{The trade-off between consumer surplus and producer surplus.} As shown in \figureref{fig:cs-and-revenue-simulation}(a), for MHR distributions (uniform and exponential), both constraints can help balance the consumer surplus and producer surplus for the whole range of regulatory intensity. However, when the distribution is strongly regular but not MHR, the monotonicity property holds for the $\gamma$-ratio constraint but fails for the $\epsilon$-difference constraint. Under this circumstance, the $\epsilon$-difference fair constraint can guarantee the property when the trade-off curve is above the X point shown in \figureref{fig:cs-and-revenue-simulation}(a), which corresponds to the proposed $\epsilon_0$ in \theoremref{thrm:gap-consumer-under-regular}. In addition, the monotonicity property of the optimal upper price can be validated in \figureref{fig:cs-and-revenue-simulation}(b) and further proves the trade-off between consumer surplus and producer surplus, according to \propositionref{prop:price-and-surplus}.
    \item \textit{Drop on total surplus.} As shown in \figureref{fig:cs-and-revenue-simulation}(a), as an anchor point moves along the curves from the upper left corner to the bottom line, the distance between the anchor point to the upper right borderline of the grey triangle becomes larger. This implies both constraints can lead to a decrease in total surplus given that the distance can be viewed as the drop of total surplus 
    from the efficient trade. Furthermore, a stricter constraint results in a larger loss on total surplus. In addition, the monotonicity property of the optimal lower price is presented in \figureref{fig:cs-and-revenue-simulation}(b) and the property can further prove the loss on total surplus, according to \propositionref{prop:price-and-surplus}. The drop on total surplus is reasonable because the perfect price discrimination achieves the maximal market efficiency and any regulatory policies attempting to avoid perfect price discrimination will inevitably harm total surplus.
    \item \textit{$\epsilon$-difference fair vs $\gamma$-ratio fair.} As shown in \figureref{fig:cs-and-revenue-simulation}, the trade-off curve of $\epsilon$-difference fair is on top of that achieved by $\gamma$-ratio fair as long as the demand distribution is strongly regular, which validates our theoretical results in \theoremref{thrm:gap-vs-ratio}.
\end{itm}

\subsection{Real-world datasets}
\subsubsection{Datasets and preprocessing}
\begin{enum}
    \item \textit{Coke and cake}. \citet{wertenbroch2002measuring} adopted Becker, DeGroot, and Marschak’s method \citep{becker1964measuring} to estimate willingness-to-pay for a can of Coca-Cola on a public beach and a piece of pound cake on a commuter ferry in Kiel, Germany. The quantity demanded are then regressed with a logistic model $S(p) \propto 1/\left(1 + \mathrm{e}^{-(a+bp)}\right)$. The fitted parameters are $a=3.94$, $b=-3.44$ for the demand of Coke and $a=4.58$, $b=-3.72$ for the demand of cake. These parameters are good estimations of the raw willingness to pay according to \citep{wertenbroch2002measuring}. We use the fitted logistic model as the demand function. As a result, the demand distributions are logistic distributions and satisfy the assumption of MHR.
    \item \textit{Elective vaccine}. \citet{slunge2015willingness} studied willingness to pay for vaccination against tick-borne encephalitis in Sweden. They asked individuals with covariate $x$ about take-up at a random price of 100, 250, 500, 750, or 1000 SEK. We follow \citep{slunge2015willingness,kallus2021fairness} and learn a logistic regression model of binary demand by appending the price variable with the other covariates, \textit{i.e.}, $D(x, p) = \sigma(\gamma^Tx+\beta p)$ where $\sigma(\cdot)$ is the logistic function. The overall demand function can be given as $S(p) = \mathbb{E}_x[D(x, p)]$.
    \item \textit{Auto loan}. The dataset records 208,085 auto loan applications received by a major online lender in the United States with loan-specific features. Following \citep{phillips2015effectiveness,ban2021personalized,luo2021distribution}, we adopt the feature selection results and consider only four features: the loan amount approved, FICO score, prime rate, and the competitor’s rate. The price $p$ of a loan is computed as the net present value of future payment minus the loan amount, \textit{i.e.}, $p=\text{Monthly Payment} \times \sum_{\tau=1}^\text{Term}(1+\text{Rate})^{-\tau}-\text{Loan Amount}$. Following \citep{luo2021distribution}, we set the rate as $0.12\%$ to estimate the monthly London interbank offered rate for the studied time period. We further fit the demand function with a logistic regression model.
\end{enum}

\begin{figure}[t]
    \centering
    \includegraphics[width=\linewidth]{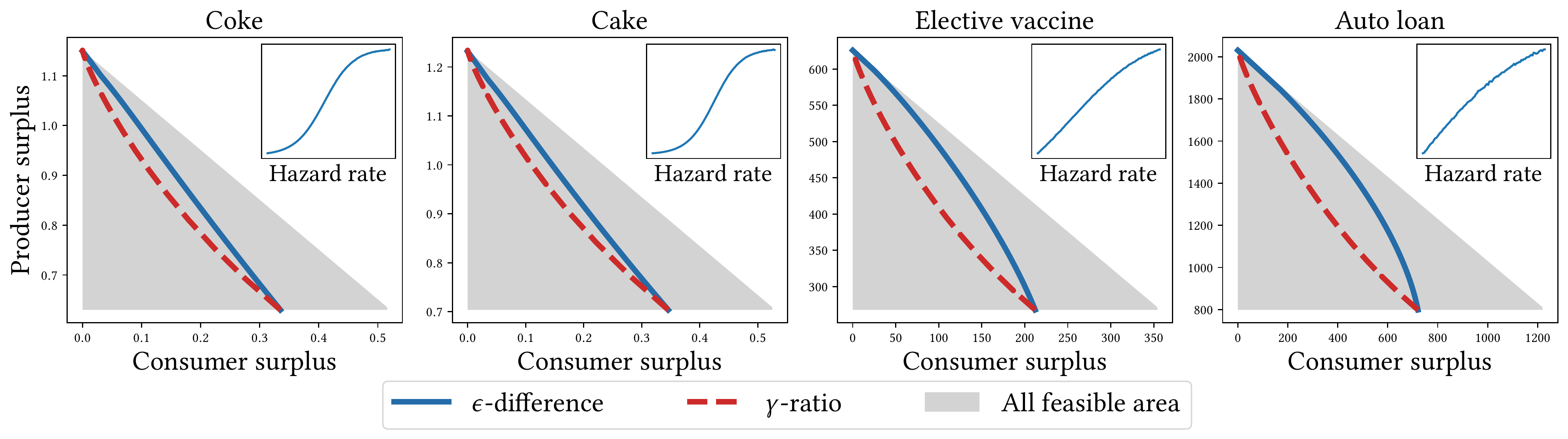}
    \caption{Experiments on real-worlds datasets. It shows trade-off curves between consumer surplus and producer surplus under the $\epsilon$-difference and $\gamma$-ratio constraints, as well as the hazard rate functions of these datasets.}
    \figurelabel{fig:cs-and-revenue-real}
\end{figure}

\subsubsection{Results and analysis} With closed-form demand functions of these datasets, we can optimize \equationref{eq:gap-constraint} and \equationref{eq:ratio-constraint} directly. We choose different $\epsilon$ and $\gamma$ and calculate the corresponding optimal lower / upper price, consumer surplus, producer surplus and total surplus accordingly.

We first analyze the hazard rate function for the demand distributions on these real-world datasets. As shown in the upper right subfigures of \figureref{fig:cs-and-revenue-real}, all of the datasets approximately satisfy the MHR condition. For coke and cake, the demand distributions are logistic and they satisfy the MHR condition by nature. For elective vaccine and auto loan, the hazard rate functions are not strictly increasing and we can find fluctuations in the figures. However, they can be considered increasing from the long-run trend.

The trade-off curve between consumer surplus and producer surplus can be found in \figureref{fig:cs-and-revenue-real}. As shown in the figure, the fluctuations on hazard rate functions do not affect the major results. Similar to simulation experiments, both constraints can help balance consumer surplus and producer surplus at the expense of total surplus. In addition, the curve of the $\gamma$-ratio constraint is on top of the $\epsilon$-difference constraint. These results match the theories we propose in \sectionref{sect:impacts}.

%% file: paragraphs/discussions.tex
\section{Conclusions} \sectionlabel{sect:discussions}
To conclude, in this paper, we propose two sound and effective policy instruments on the range of personalized prices and study their impact on consumer surplus, producer surplus, and social welfare. For common demand distributions, both constraints can help balance consumer surplus and producer surplus at the expense of total surplus. In addition, the $\epsilon$-difference constraint has better performance on the trade-off between consumer surplus and producer surplus  while the $\gamma$-ratio constraint is more suitable for designing policies. As a result, the two constraints could be adopted in different applications in practice.

%% file: paragraphs/acknowledgments.tex
\section{Acknowledgments}
This work was supported in part by National Natural Science Foundation of China (No. 62141607, U1936219, 72171131), National Key R\&D Program of China (No. 2020AAA0106300), Beijing Academy of Artificial Intelligence (BAAI), the Tsinghua University Initiative Scientific Research Grant (No. 2019THZWJC11), and Technology and Innovation Major Project of the Ministry of Science and Technology of China under Grant 2020AAA0108400 and 2020AAA01084020108403.

%% file: paragraphs/appendix.tex
\section{Omitted proofs}
\subsection{Proof of \texorpdfstring{\propositionref{prop:mhr-and-regular}}{Proposition}}
\begin{proof}
    Suppose $F$ is a monotone hazard rate distribution, which means $h(v)$ is non-decreasing. As a result, $w(v) = v - \frac{1}{h(v)}$ is strictly increasing. In addition,
    $$
    \lim_{v \rightarrow U} w(v) \begin{cases}
        \ge \lim_{v \rightarrow \infty} v - \frac{1}{h(0)} = \infty, & \text{if } U = \infty, \\
        = U - \frac{S(U)}{f(U)} = U > k & \text{otherwise}.
    \end{cases}
    $$
    As a result, $F$ is $k$-strongly regular.
\end{proof}

\subsection{Proof of \texorpdfstring{\propositionref{prop:gap-existence}}{Proposition}}
To prove the proposition, we need the following lemma.
\begin{lemma} \lemmalabel{lemma:gap-strong-regular}
    If $F$ is a $c$-strongly regular distribution, then $\forall 0 \le \epsilon < U$, $r_{\epsilon}(v) \triangleq v - \epsilon - S(v)/f(v - \epsilon)$ is strictly increasing. In addition, $\lim_{v \rightarrow U}r_{\epsilon}(v) > 0$.
\end{lemma}
\begin{proof}
    The results for $r_0(v)$ hold by definition of strong regularity. Now consider $\epsilon > 0$. We express $r_{\epsilon}$ with hazard functions.
    $$
    r_{\epsilon}(v) = v - \epsilon - \frac{S(v)}{S(v - \epsilon)}\cdot \frac{S(v - \epsilon)}{f(v - \epsilon)} =  v - \epsilon - \frac{\exp\left(-\int_{v-\epsilon}^vh(t)\mathrm{d}t\right)}{h(v - \epsilon)}.
    $$
    Because $F$ is strongly regular, $w(v) = v - 1 / h(v)$ is strictly increasing, which implies $\frac{\mathrm{d} w}{\mathrm{d} v} = 1 + \frac{h'(v)}{h^2(v)} \ge 0$. Hence, $h'(v) \ge -h^2(v)$. Then we calculate the derivative of $r_{\epsilon}$,
    $$
    \begin{aligned}
        & \frac{\mathrm{d} r_{\epsilon}}{\mathrm{d} v} \\
        = & 1 - \frac{\exp\left(-\int_{v-\epsilon}^v h(t)\mathrm{d}t\right)}{h^2(v-\epsilon)}\left(h^2(v-\epsilon) - h(v)h(v-\epsilon) - h'(v-\epsilon)\right) \\
        \ge & 1 + \exp\left(-\int_{v-\epsilon}^v h(t)\mathrm{d}t\right)\left(\frac{h(v)}{h(v-\epsilon)} - 2\right).
    \end{aligned}
    $$
    Because $w(v)$ is strictly increasing, then $\forall t \in (v - \epsilon, v]$, $w(t) > w(v - \epsilon)$, which implies that $
    t - \frac{1}{h(t)} > v - \epsilon - \frac{1}{h(v - \epsilon)}$. Hence, $h(t) > \frac{1}{t - v + \epsilon + 1/h(v-\epsilon)}$.
    As a result, $\frac{h(v)}{h(v - \epsilon)} > \frac{1}{1 + h(v-\epsilon)\epsilon}$ and
    $$
    \begin{aligned}
        & \int_{v-\epsilon}^v h(t)\mathrm{d} t > \int_{v-\epsilon}^v \frac{1}{t - v + \epsilon + \frac{1}{h(v-\epsilon)}} \mathrm{d} t \\
        = & \left.\ln\left(t - v + \epsilon + \frac{1}{h(v-\epsilon)}\right)\right|_{t=v - \epsilon}^v = \ln \left(1 + h(v-\epsilon)\epsilon\right).
    \end{aligned}
    $$
    As a result,
    $$
    \begin{aligned}
        & \exp\left(\int_{v-\epsilon}^v h(t)\mathrm{d}t\right) \cdot \frac{\mathrm{d} r_{\epsilon}}{\mathrm{d} v} \ge \exp\left(\int_{v-\epsilon}^v h(t)\mathrm{d}t\right) + \frac{h(v)}{h(v-\epsilon)} - 2 \\
        > & 1 + h(v-\epsilon)\epsilon + \frac{1}{1 + h(v-\epsilon)\epsilon} - 2 \ge 0,
    \end{aligned}
    $$
    which implies $r_{\epsilon}(v)$ is strictly increasing. Now consider the limits of $r_{\epsilon}(v)$ when $v \rightarrow U$.
    $$
    \begin{aligned}
        & \quad \, \lim_{v \rightarrow U}r_{\epsilon}(v) = \lim_{v \rightarrow U} v - \epsilon - \frac{S(v)}{f(v - \epsilon)} \\
        & \begin{cases}
            > \lim_{v \rightarrow \infty} v - \epsilon - \frac{S(v - \epsilon)}{f(v - \epsilon)} = \lim_{v \rightarrow \infty} w(v) > 0, & \text{if } U =\infty, \\
            = U - \epsilon > 0, & \text{otherwise.}
        \end{cases}
    \end{aligned}
    $$
\end{proof}
\begin{proof}[Proof of \propositionref{prop:gap-existence}]
    It is obvious that $p^*_u(\epsilon) - p^*_l(\epsilon) = \epsilon$. Hence, $p_l^*(\epsilon) = \arg \max_{p_l} \PS(p_l, p_l + \epsilon)$. The derivative of $\PS(p_l, p_l+\epsilon)$ is
    $$
        G_l(p_l, \epsilon) \triangleq \frac{\mathrm{d} \PS}{\mathrm{d} p_l} = S(p_l + \epsilon) - p_lf(p_l) = f(p_l)\left(\frac{S(p_l + \epsilon)}{f(p_l)} - p_l\right).
    $$
    According to \lemmaref{lemma:gap-strong-regular}, $S(p_l + \epsilon)/f(p_l) - p_l = -r_{\epsilon}(p_l + \epsilon)$ is strictly decreasing \textit{w.r.t.} $p_l$. In addition,
    $$
    \begin{aligned}
        \lim_{p_l \rightarrow 0} \left(\frac{S(p_l + \epsilon)}{f(p_l)} - p_l\right) & = \frac{S(\epsilon)}{f(0)} > 0, \\
        \lim_{p_l \rightarrow (U - \epsilon)} \left(\frac{S(p_l + \epsilon)}{f(p_l)} - p_l\right) & = - \lim_{v \rightarrow U} r_{\epsilon}(v) < 0 \\
    \end{aligned}
    $$
    which implies the solution to $\mathrm{d} \PS / \mathrm{d} p_l = 0$ exists and is unique. As a result, for all $0 \le \epsilon < U$, the solution $\left(p_l^*(\epsilon), p_u^*(\epsilon)\right)$ exists and is unique.

    In addition, $G_l(p_l^*, \epsilon) = 0$ is the implicit function that determines the relationship between $p_l^*$ and $\epsilon$ and $G_l(p_l^*, \epsilon)$ is differentiable. Notice that according to \lemmaref{lemma:gap-strong-regular}, we have
    $$
    \frac{\mathrm{d} r_{\epsilon}}{\mathrm{d} v} = \frac{f^2(v-\epsilon) + f(v)f(v - \epsilon) + S(v)f'(v-\epsilon)}{f^2(v-\epsilon)} > 0,
    $$
    Hence,
    $$
    \begin{aligned}
        & \frac{\partial G_l}{\partial p_l}\left(p_l^*, \epsilon\right) = -f(p_l^*+\epsilon) - f(p_l^*) - p_l^*f'(p_l^*) \\
        = & -\left(f(p_l^* + \epsilon) + f(p_l^*) + \frac{S(p_l^* + \epsilon)}{f(p_l^*)}f'(p_l^*)\right) \\
        = & -f(p_l^*) \cdot \left.\frac{\mathrm{d} r_{\epsilon}}{\mathrm{d}v}\right|_{v = p_l^*+\epsilon} < 0.
    \end{aligned}
    $$
    According to the implicit function theorem \citep[Theorem 1.3.1]{krantz2012implicit}, $p_l^*(\epsilon)$ is differentiable, which implies that $p_u^*(\epsilon)$ is also differentiable.
\end{proof}

\subsection{Proof of \texorpdfstring{\propositionref{prop:ratio-existence}}{Proposition}}
To prove the proposition, we need the following lemma.
\begin{lemma} \lemmalabel{lemma:ratio-strong-regular}
    If $F$ is a $c$-strongly regular distribution, then $\forall \gamma \ge 1$, $z_{\gamma}(v) \triangleq v  - \frac{\gamma S(\gamma v)}{f(v)}$ is strictly increasing. In addition, $\lim_{\gamma v \rightarrow U}z_{\gamma}(v) > 0$.
\end{lemma}
\begin{proof}
    The results for $z_1(v)$ hold by definition, now consider $\gamma > 1$. We express $z_{\gamma}$ with hazard functions.
    $$
    z_{\gamma}(v)  = v - \gamma \cdot \frac{ S(\gamma v)}{S(v)} \cdot \frac{S(v)}{f(v)} = v - \gamma \frac{\exp\left(-\int_{v}^{\gamma v}h(t)\mathrm{d}t\right)}{h(v)}.
    $$
    The derivative of $z_{\gamma}$ is given by
    $$
    \begin{aligned}
        \frac{\mathrm{d} z_{\gamma}}{\mathrm{d} v} & = 1 - \gamma \cdot \frac{\exp\left(-\int_{v}^{\gamma v}h(t)\mathrm{d}t\right)}{h^2(v)}\left(h^2(v)-\gamma h(v)h(\gamma v) - h'(v)\right) \\
        & \ge 1 + \gamma \cdot \exp\left(-\int_{v}^{\gamma v}h(t)\mathrm{d}t\right)\left(\frac{\gamma h(\gamma v)}{h(v)} - 2\right).
    \end{aligned}
    $$
    Similar with the proof of \lemmaref{lemma:gap-strong-regular}, the last inequality is due to the monotonicity of $w(v)$. We also have $\forall t \in (v, \gamma v]$, $h(t) > \frac{1}{t - v + 1 / h(v)}$. As a result, $\frac{h(\gamma v)}{h(v)} > \frac{1}{1 + (\gamma - 1)vh(v)}$ and
    $$
    \begin{aligned}
        & \int_v^{\gamma v}h(t)\mathrm{d}t > \int_v^{\gamma v}\frac{1}{t-v + \frac{1}{h(v)}}\mathrm{d} t \\
        = & \left.\ln \left(t - v + \frac{1}{h(v)}\right)\right|_{t=v}^{\gamma v} = \ln\left(1+(\gamma-1)vh(v)\right).
    \end{aligned}
    $$
    As a result,
    $$
    \begin{aligned}
        & \frac{1}{\gamma} \exp\left(\int_{v}^{\gamma v}h(t)\mathrm{d}t\right) \cdot \frac{\mathrm{d} z_{\gamma}}{\mathrm{d} v} \ge \frac{1}{\gamma}\exp\left(\int_{v}^{\gamma v}h(t)\mathrm{d}t\right) + \frac{\gamma h(\gamma v)}{h(v)} - 2 \\
        > & \frac{1+(\gamma-1)vh(v)}{\gamma} + \frac{\gamma}{1+(\gamma-1)vh(v)} - 2 \ge 0,
    \end{aligned}
    $$
    which implies $z_{\gamma}(v)$ is strictly increasing. Now consider the limits of $z_{\gamma}(v)$ when $\gamma \rightarrow U$. If $U$ is finite, then $\lim_{\gamma v \rightarrow U} z_{\gamma}(v) = U/\gamma > 0$. When $U = \infty$, because $\lim_{v\rightarrow \infty}w(v) > 0$, we have $\lim_{v\rightarrow \infty}vh(v) > 1$. Then when $v \rightarrow \infty$,
    $$
    \begin{aligned}
        \frac{\gamma S(\gamma v)}{S(v)} & = \gamma \exp\left(-\int_{v}^{\gamma v}h(t)\mathrm{d}t\right) < \frac{\gamma}{1 + (\gamma - 1)vh(v)} \\
        & < \frac{\gamma}{1 + (\gamma - 1) \cdot 1} = 1.
    \end{aligned}
    $$
    Hence, when $v \rightarrow \infty$,
    $$
    z_{\gamma}(v) = v - \frac{\gamma S(\gamma v)}{S(v)} \cdot \frac{1}{h(v)} > v - \frac{1}{h(v)} = w(v) > 0.
    $$
\end{proof}

\begin{proof}[Proof of \propositionref{prop:ratio-existence}]
    It is obvious that $q_u^*(\gamma) = \gamma q_l^*(\gamma)$. Hence,
    $$
    q_l^*(\gamma) = \arg \max_{q_l} \PS(q_l, \gamma q_l).
    $$
    The derivative of $\PS(q_l, \gamma q_l)$ is
    $$
    H_l(q_l, \gamma) \triangleq \frac{\mathrm{d} \PS}{\mathrm{d}q_l} = \gamma S(\gamma q_l) - q_lf(q_l) = f(q_l)\left(\frac{\gamma S(\gamma q_l)}{f(q_l)}-q_l\right).
    $$
    According to \lemmaref{lemma:ratio-strong-regular}, $\left(\frac{\gamma S(\gamma q_l)}{f(q_l)}-q_l\right)$ is strictly decreasing. In addition,
    $$
    \begin{aligned}
        \lim_{q_l \rightarrow 0}\left(\frac{\gamma S(\gamma q_l)}{f(q_l)}-q_l\right) & = \frac{\gamma S(0)}{f(0)} > 0, \\
        \lim_{\gamma q_l \rightarrow U}\left(\frac{\gamma S(\gamma q_l)}{f(q_l)}-q_l\right) & = -\lim_{\gamma q_l \rightarrow U} z_{\gamma}(q_l) < 0,
    \end{aligned}
    $$
    which implies the solution to $\frac{\mathrm{d} \PS}{\mathrm{d} q_l} = 0$ exists and is unique. As a result, for all $\gamma \ge 1$, the solution $\left(q_l^*(\gamma), q_u^*(\gamma)\right)$ exists and unique.

    In addition, $H_l(q_l^*, \gamma) = 0$ is the implicit function that determines the relationship between $q_l^*$ and $\gamma$ and $H_l(q_l^*, \gamma)$ is continuously differentiable. Notice that according to \lemmaref{lemma:ratio-strong-regular},
    $$
    \frac{\mathrm{d} z_{\gamma}}{\mathrm{d} v} = \frac{f^2(v) + \gamma^2 f(v)f(\gamma v) + \gamma S(\gamma v)f'(v)}{f^2(v)} > 0,
    $$
    Hence,
    $$
    \begin{aligned}
        \frac{\partial H_l}{\partial q_l}\left(q_l^*, \gamma\right) = & -\gamma^2f(\gamma q_l^*) - f(q_l^*) - q_l^* f'(q_l^*) \\
        = & -\left(\gamma^2 f(\gamma q_l^*) + f(q_l^*) + \frac{\gamma S(\gamma q_l^*)}{f(q_l^*)}f'(q_l^*)\right) \\
        = & -f(q_l^*) \cdot \left.\frac{\mathrm{d} z_{\gamma}}{\mathrm{d}v}\right|_{v = q_l^*} < 0.
    \end{aligned}
    $$
    According to the implicit function theorem \citep[Theorem 1.3.1]{krantz2012implicit}, $q_l^*(\epsilon)$ is differentiable, which implies that $q_u^*(\epsilon)$ is also differentiable.
\end{proof}

\subsection{Proof of \texorpdfstring{\propositionref{prop:price-and-surplus}}{Proposition}}
\begin{proof}
    On the one hand, when $c < p_l < U$,
    $$
    \frac{\partial \TS (p_l, p_u)}{\partial p_l} = -(p_l-c)f(p_l) < 0.
    $$
    Hence total surplus is strictly decreasing \textit{w.r.t.} $p_l$. On the other hand, when $0 < p_u < U$,
    $$
    \frac{\partial \CS (p_l, p_u)}{\partial p_u} = - (p_u - p_u)f(p_u) + \int_{p_u}^U -f(v) \mathrm{d} v = -S(p_u) < 0,
    $$
    which implies consumer surplus is strictly decreasing \textit{w.r.t.} $p_u$.
\end{proof}

\subsection{Proof of \texorpdfstring{\theoremref{thrm:gap-increasing-upper-price}}{Theorem}}
\begin{proof}
    $\forall 0 < \epsilon < U$, it is obvious that $p^*_u(\epsilon) - p^*_l(\epsilon) = \epsilon$. The optimal $p_u^*$ is given by $p_u^*(\epsilon) = \arg \max_{p_u} \PS(p_u - \epsilon, p_u)$. The derivative of $\PS(p_u - \epsilon, p_u)$ is
    $$
    G_u(p_u, \epsilon) \triangleq \frac{\mathrm{d} \PS}{\mathrm{d} p_u}  = S(p_u) - (p_u-\epsilon)f(p_u-\epsilon).
    $$
    According to \propositionref{prop:gap-existence}, $G_u(p_u^*, \epsilon) = 0$ is the implicit function that determines the relationship between $p_u^*$ and $\epsilon$. According to \citep[Theorem 1.3.1]{krantz2012implicit}, the derivative of $p_u^*(\epsilon)$ is given by
    $$
        \frac{\mathrm{d}p_u^*}{\mathrm{d} \epsilon} = -\frac{\frac{\partial G_u}{\partial \epsilon}\left(p_u^*, \epsilon\right)}{\frac{\partial G_u}{\partial p_u}\left(p_u^*, \epsilon\right)} = \frac{f(p_u^*-\epsilon)+(p_u^*-\epsilon)f'(p_u^*-\epsilon)}{f(p_u^*) + f(p_u^* - \epsilon) + (p_u^*-\epsilon)f'(p_u^*-\epsilon)}.
    $$
    Because $F$ is a monotone hazard rate distribution, $h(v)$ is non-decreasing, which implies, $\frac{\mathrm{d} h}{\mathrm{d} v} = \frac{f'(v)S(v) + f^2(v)}{S^2(v)} \ge 0$. As a result, $f'(p_u^*-\epsilon) \ge -f^2(p_u^*-\epsilon)/S(p_u^*-\epsilon)$. Hence,
    $$
    \begin{aligned}
        & f(p_u^*-\epsilon)+(p_u^*-\epsilon)f'(p_u^*-\epsilon) \\
        \ge & f(p_u^*-\epsilon) - (p_u^*-\epsilon)\cdot\frac{f(p_u^*-\epsilon)f(p_u^*-\epsilon)}{S(p_u^*-\epsilon)} \\
        = & f(p_u^*-\epsilon)\left(1-\frac{S(p_u^*)}{S(p_u^* - \epsilon)}\right) > 0.
    \end{aligned}
    $$
    As a result, $\mathrm{d}p_u^*/\mathrm{d} \epsilon > 0$ and $p_u^*$ is strictly increasing \textit{w.r.t.} $\epsilon$. In addition, according to \propositionref{prop:price-and-surplus}, $\CSGap(\epsilon)$ is strictly decreasing.
\end{proof}

\subsection{Proof of \texorpdfstring{\theoremref{thrm:gap-consumer-under-regular}}{Theorem}}
\begin{proof}
    $\forall 0 < \epsilon < U$, consider the numerator of $\mathrm{d} p_u^*(\epsilon) / \mathrm{d} \epsilon$,
    $$
    \begin{aligned}
        & f(p_u^*-\epsilon)+(p_u^*-\epsilon)f'(p_u^*-\epsilon) = f(p_l^*) + p_l^*f'(p_l^*) \\
        = & h(p_l^*)\exp\left(-\int_0^vh(t)\mathrm{d}t\right) + p_l^*\left(h(p_l^*)\exp\left(-\int_0^vh(t)\mathrm{d}t\right)\right)'\\
        = & \exp \left(-\int_0^{p_l^*} h(t)\mathrm{d}t\right)\left(h(p_l^*) + p_l^*h'(p_l^*)-p_l^*h^2(p_l^*)\right) \\
        \ge & h(p_l^*)\exp \left(-\int_0^{p_l^*} h(t)\mathrm{d}t\right)\left(1 - 2p_l^*h(p_l^*)\right).
    \end{aligned}
    $$
    The last inequality is due to the monotonicity of $w(v) = v - 1 / h(v)$, which means $w'(v) = 1 + h'(v)/h^2(v) \ge 0$. The relationship between $p_l^*$ and $\epsilon$ is given by $G_l(p_l^*, \epsilon) = S(p_l^* + \epsilon) - p_l^*f(p_l^*) = 0$, \textit{i.e.},
    $$
    \exp\left(-\int_{p_l^*}^{p_l^*+\epsilon}h(v)\mathrm{d} v\right) - p_l^*h(p_l^*) = 0.
    $$
    Because $w(v)$ is strictly increasing, then $\forall v \in (p_l^*, p_l^* + \epsilon]$, $w(v) > w(p_l^*)$, which implies that $
    v - 1/h(v) > p_l^* - 1/h(p_l^*)$. Hence,
    $$
    h(v) > \frac{1}{v - p_l^* + 1/h(p_l^*)}.
    $$
    As a result,
    $$
    \int_{p_l^*}^{p_l^* + \epsilon} h(v)\mathrm{d}v > \int_{p_l^*}^{p_l^* + \epsilon} \frac{1}{v - p_l^* + \frac{1}{h(p_l^*)}} \mathrm{d} v = \ln \left(1 + h(p_l^*)\epsilon\right).
    $$
    This implies
    $$
    p_l^*h(p_l^*) = \exp\left(-\int_{p_l^*}^{p_l^*+\epsilon}h(v)\mathrm{d} v\right) < \frac{1}{1 + h(p_l^*)\epsilon},
    $$
    and $h(p_l^*) < \frac{-p_l^* + \sqrt{(p_l^*)^2 + 4p_l^*\epsilon}}{2p_l^*\epsilon}$.
    Let $t(\epsilon) = p_l^*(\epsilon) / \epsilon$. Then when $\epsilon > \epsilon_0$, because of the monotonicity of $p_l^*$ according to \theoremref{thrm:gap-decreasing-lower-price}, $t(\epsilon) < p_l^*(\epsilon_0) / \epsilon_0 = 1 / 2$. Therefore, when $\epsilon > \epsilon_0$,
    $$
    p_l^*h(p_l^*) < \frac{1}{2}\left(-t+\sqrt{t^2+4t}\right) < \frac{1}{2}\left(-\frac{1}{2}+\sqrt{\left(\frac{1}{2}\right)^2+4\cdot\frac{1}{2}}\right) = \frac{1}{2}.
    $$
    Hence,
    $$
    \begin{aligned}
        & f(p_u^*-\epsilon)+(p_u^*-\epsilon)f'(p_u^*-\epsilon) \\
        \ge & h(p_l^*)\exp \left(-\int_0^{p_l^*} h(t)\mathrm{d}t\right)\left(1 - 2p_l^*h(p_l^*)\right) \\
        > & h(p_l^*)\exp \left(-\int_0^{p_l^*} h(t)\mathrm{d}t\right)\left(1 - 2 \cdot \frac{1}{2}\right) = 0,
    \end{aligned}
    $$
    which implies $\frac{\mathrm{d}p_u^*}{\mathrm{d} \epsilon} > 0$ and $p_u^*$ is increasing when $\epsilon > \epsilon_0$. As a result, according to \propositionref{prop:price-and-surplus}, $\CSGap(\epsilon)$ is strictly decreasing when $\epsilon > \epsilon_0$.
\end{proof}

\subsection{Proof of \texorpdfstring{\theoremref{thrm:ratio-monotone-upper-price} and \theoremref{thrm:ratio-monotone-lower-price}}{Theorem}}
\begin{proof}
    We first prove the non-increasing property of $q_l^*$. According to \propositionref{prop:ratio-existence}, $H_l(q_l, \gamma) = 0$ is the implicit function that determines the relationship between $q_l^*$ and $\gamma$. Notice that when $\gamma = 1$, the optimal solution $q_l^*(1)$ and $q_u^*(1)$ should satisfy $q_l^*(1) = q_u^*(1)$ and $S(q_u^*(1)) = q_u^*(1)f(q_u^*(1))$. According to \citep[Theorem 1.3.1]{krantz2012implicit}, the derivative of $q_l^*(\gamma)$ is given by
    $$
    \begin{aligned}
        \frac{\mathrm{d} q_l^*}{\mathrm{d} \gamma} & = - \frac{\frac{\partial H_l}{\mathrm{d} \gamma}\left(q_l^*, \gamma\right)}{\frac{\partial H_l}{\mathrm{d} q_l}\left(q_l^*, \gamma\right)} = \frac{S(\gamma q_l^*) - \gamma q_l^* f(\gamma q_l^*)}{\gamma^2f(\gamma q_l^*) + f(q_l^*) + q_l^* f'(q_l^*)} \\
        & = \frac{S(q_u^*) - q_u^* f(q_u^*)}{\gamma^2f(\gamma q_l^*) + f(q_l^*) + q_l^* f'(q_l^*)}.
    \end{aligned}
    $$
    According to the proof of \propositionref{prop:ratio-existence}, the denominator of the equation above is greater than $0$. Now consider the numerator.
    
    Suppose there exists $\gamma_0 > 1$ such that $q_u^*(\gamma_0) < q_u^*(1)$. Because of the differentiability of $q_u^*(\gamma)$ according to \propositionref{prop:ratio-existence}, there exist a range $[\gamma_1, \gamma_2] \subseteq [1, \gamma_0]$ such that $q_u^*(\gamma)$ is non-increasing when $\gamma \in [\gamma_1, \gamma_2]$ and $q_u^*(\gamma_1) < q_u^*(1)$. Hence, $\forall \gamma \in [\gamma_1, \gamma_2]$, $q_u^*(\gamma) < q_u^*(1)$. Because $w(v) = v - S(v) / f(v)$ is strictly increasing,
    $$
    q_u^*(\gamma) - \frac{S(q_u^*(\gamma))}{f(q_u^*(\gamma))} < q_u^*(1) - \frac{S(q_u^*(1))}{f(q_u^*(1))} = 0,
    $$
    which means $S(q_u^*(\gamma)) - q_u^*(\gamma)f(q_u^*(\gamma)) > 0$ and $\frac{\mathrm{d} q_l^*}{\gamma} > 0$. Hence $q_l^*(\gamma)$ is strictly increasing when $\gamma \in [\gamma_1, \gamma_2]$. Because $q_u^*(\gamma)$ is non-increasing when $\gamma \in [\gamma_1, \gamma_2]$, $\gamma = q_u^*(\gamma) / q_l^*(\gamma)$ is non-increasing when $\gamma \in [\gamma_1, \gamma_2]$, which results in a contradiction. Hence $\forall \gamma \ge 1$, $q_u^*(\gamma) \ge q_u^*(1)$. As a result, $S(q_u^*) - q_u^*f(q_u^*) \le 0$, which implies $q_l^*(\gamma)$ is non-increasing \textit{w.r.t.} $\gamma$.

    Next we prove the non-decreasing property of $q_u^*$. Similarly, the optimal $q_u^*$ is given by
    $$
    q_u^*(\gamma) = \arg \max_{q_u} \PS(q_u / \gamma, q_u).
    $$
    The derivative of $\PS(q_u / \gamma, q_u)$ is
    $$
    H_u(q_l, \gamma) \triangleq \frac{\mathrm{d} \PS}{\mathrm{d} q_u} = S(q_u) - \frac{q_u f(q_u / \gamma)}{\gamma^2}.
    $$
    According to \propositionref{prop:ratio-existence}, $H_u(q_u^*, \gamma)=0$ is the implicit function that determines the relationship between $q_u^*$ and $\gamma$. According to \citep[Theorem 1.3.1]{krantz2012implicit}, the derivative of $q_u^*(\gamma)$ if given by
    $$
    \begin{aligned}
        \frac{\mathrm{d} q_u^*}{\mathrm{d} \gamma} & = - \frac{\frac{\partial H_u}{\mathrm{d} \gamma}\left(q_u^*, \gamma\right)}{\frac{\partial H_u}{\mathrm{d} q_u}\left(q_u^*, \gamma\right)} = \frac{\left(\frac{q_u^*}{\gamma}\right)^2f'\left(\frac{q_u^*}{\gamma}\right)+2\frac{q_u^*}{\gamma} f\left(\frac{q_u^*}{\gamma}\right)}{\gamma^2f(q_u^*)+f\left(\frac{q_u^*}{\gamma}\right)+\frac{q_u^*}{\gamma} f'\left(\frac{q_u^*}{\gamma}\right)} \\
        & = \frac{\left(q_l^*\right)^2f'(q_l^*) + 2q_l^*f(q_l^*)}{\gamma^2f(\gamma q_l^*)+f(q_l^*) + q_l^*f'(q_l^*)}.
    \end{aligned}
    $$
    According to the proof of \propositionref{prop:ratio-existence}, the denominator of the equation above is greater than $0$. Now consider the numerator. Because $w(v) = v - S(v) / f(v)$ is strictly increasing,
    $$
    w'(v) = \frac{2f^2(v) + S(v)f'(v)}{f^2(v)} \ge 0.
    $$
    Hence $f'(q_l^*) \ge -2f^2(q_l^*)/S(q_l^*)$, and
    $$
    \begin{aligned}
        \left(q_l^*\right)^2f'(q_l^*) + 2q_l^*f(q_l^*) \ge \frac{2q_l^*f^2(q_l^*)}{S(q_l^*)}\left(\frac{S(q_l^*)}{f(q_l^*)} - q_l^*\right).
    \end{aligned}
    $$
    Because $q_l^*(\gamma)$ is non-increasing, $q_l^*(\gamma) \le q_l^*(1)$. For the monotonicity of $w(v) = v- S(v) / f(v)$,
    $$
    \frac{S(q_l^*(\gamma))}{f(q_l^*(\gamma))} - q_l^*(\gamma) = -w(q_l^*(\gamma)) \ge -w(q_l^*(1)) = 0.
    $$
    As a result, $\frac{\mathrm{d} q_u^*}{\mathrm{d} \gamma} \ge 0$ and $q_u^*$ is non-decreasing.

    Then We could prove the strict monotonicity of $q_l^*$ and $q_u^*$. Suppose there exists $\gamma_0 > 1$ such that $q_u^*(\gamma_0) = q_u^*(1)$. Because the non-decreasing property, $\forall \gamma \in [1, \gamma_0]$, $q_u^*(\gamma) = q_u^*(1)$. Then $\forall \gamma \in [1, \gamma_0]$, $\mathrm{d} q_l^* / \mathrm{d} \gamma = 0$, which means $q_l^*(\gamma) = q_l^*(1)$. Therefore, $\gamma_0 = q_u^*(\gamma_0) / q_l^*(\gamma_0) = q_u^*(1) / q_l^*(1) = 1$, which results in a contradiction. As a result, $\forall \gamma > 1$, $q_u^*(\gamma) > q_u^*(1)$. Therefore, $\forall \gamma > 1$, $\mathrm{d} q_l^* / \mathrm{d} \gamma < 0$, and $q_l^*(\gamma)$ is strictly decreasing. Similarly, $\forall \gamma > 1$, $q_l^*(\gamma) < q_l^*(1)$, which implies $\mathrm{d} q_u^* / \mathrm{d} \gamma > 0$ and $q_u^*(\gamma)$ is strictly increasing.

    Finally, according to \propositionref{prop:price-and-surplus}, $\CSRatio(\gamma)$ is strictly decreasing and $\TSRatio(\gamma)$ is strictly increasing \textit{w.r.t.} $\gamma$. 
\end{proof}

\subsection{Proof of \texorpdfstring{\theoremref{thrm:gap-decreasing-lower-price}}{Theorem}}
\begin{proof}
    $\forall 0 < \epsilon < U$, according to \propositionref{prop:gap-existence}, $G_l(p_l^*, \epsilon) = 0$ is the implicit function that determines the relationship between $p_l^*$ and $\epsilon$. According to \citep[Theorem 1.3.1]{krantz2012implicit}, the derivative of $p_l^*(\epsilon)$ is given by
    $$
    \frac{\mathrm{d}p_l^*}{\mathrm{d} \epsilon} = -\frac{\frac{\partial G_l}{\partial \epsilon}\left(p_l^*, \epsilon\right)}{\frac{\partial G_l}{\partial p_l}\left(p_l^*, \epsilon\right)} = -\frac{f(p_l^* + \epsilon)}{f(p_l^* + \epsilon) + f(p_l^*) + p_l^*f'(p_l^*)}.
    $$
    According to the proof of \propositionref{prop:gap-existence},
    $$
    f(p_l^* + \epsilon) + f(p_l^*) + p_l^*f'(p_l^*) = -\frac{\partial G_l}{\partial p_l}\left(p_l^*, \epsilon\right) > 0.
    $$
    Hence, $\frac{\mathrm{d}p_l^*}{\mathrm{d} \epsilon} < 0$ and $p_l^*(\epsilon)$ is strictly decreasing. As a result, according to \propositionref{prop:price-and-surplus}, $\TSGap(\epsilon)$ is strictly increasing \textit{w.r.t.} $\epsilon$.
\end{proof}

\subsection{Proof of \texorpdfstring{\theoremref{thrm:gap-vs-ratio}}{Theorem}}
\begin{proof}
    Prove the theorem by contradiction. Suppose $\TSGap(\epsilon) < \TSRatio(\gamma)$. According to \propositionref{prop:price-and-surplus}, we have $q_l^*(\gamma) < p_l^*(\epsilon)$, $q_u^*(\gamma) = p_u^*(\epsilon)$. Hence, let $\gamma' \triangleq p_u^*(\epsilon)/p_l^*(\epsilon) < q_u^*(\gamma)/q_l^*(\gamma) = \gamma$.
    
    On the one hand, by the strict monotonicity of $q_l^*$ and $q_u^*$ suggested by \theoremref{thrm:ratio-monotone-upper-price} and \theoremref{thrm:ratio-monotone-lower-price}, $q_l^*(\gamma') > q_l^*(\gamma)$ and $q_u^*(\gamma') < q_u^*(\gamma)$. According to \propositionref{prop:ratio-existence}, $(q_l^*(\gamma'),q_u^*(\gamma'))$ is the solution to \equationref{eq:ratio-constraint} and is unique, we have $\PS(q_l^*(\gamma'), q_u^*(\gamma')) > \PS(p_l^*(\epsilon), p_u^*(\epsilon))$. On the other hand,
    $$
    q_u^*(\gamma') - q_l^*(\gamma') = (\gamma' - 1)q_l^*(\gamma') < (\gamma' - 1)p_l^*(\epsilon) = p_u^*(\epsilon) - p_l^*(\epsilon) = \epsilon.
    $$
    As a result,
    $$
    \PS(q_l^*(\gamma'), q_u^*(\gamma')) \le \PSGap(q_u^*(\gamma') - q_l^*(\gamma')) \le \PS(p_l^*(\epsilon), p_u^*(\epsilon)),
    $$
    which leads to a contradiction. To conclude, we have $\TSGap(\epsilon) \ge \TSRatio(\gamma)$. In addition, because $\CSGap(\epsilon) = \CSRatio(\gamma)$, we have $\PSGap(\epsilon) \ge \PSRatio(\gamma)$.
\end{proof}

\subsection{Proof of \texorpdfstring{\propositionref{prop:hold-valuation-assumption}}{Proposition}}
\begin{proof}
    The conclusion is obvious for monotone hazard rate distributions. Now if $F$ is $c$-strongly regular, $\tilde{w}(v) = v - \frac{S(v + c)}{f(v + c)}$ is obvious strictly increasing. In addition,
    $$
    \lim_{v \rightarrow U} \tilde{w}(v) = \lim_{v \rightarrow U} \left(v + c - \frac{S(v + c)}{f(v + c)}\right) - c > c - c = 0.
    $$
    Hence $\tilde{F}$ is strongly regular.
\end{proof}